%% file: main-arxiv.tex
\newif\ifDRAFT 
\DRAFTfalse
\DRAFTtrue 

\documentclass[11pt]{article}
\usepackage[margin=1in]{geometry}
\usepackage{graphicx} 
\usepackage[framemethod=TikZ]{mdframed}
\usepackage{amsfonts, amsmath, amssymb, amsthm}
\allowdisplaybreaks
\usepackage{mathrsfs}  
\usepackage{algorithm}
\usepackage[noend]{algpseudocode}
\usepackage{hyperref}
\usepackage{wrapfig}
\usepackage{color}
\usepackage{array}
\usepackage{cleveref}
\hypersetup{colorlinks=false}
\usepackage[inkscapeformat=pdf]{svg}
\usepackage{multirow}
\usepackage{mathtools}
\usepackage{enumitem}
\usepackage{thm-restate}
\usepackage{tikz}

\ifDRAFT
    \usepackage{lineno}
\fi

\theoremstyle{plain}
\newtheorem{theorem}{Theorem}[section]
\newtheorem{lemma}[theorem]{Lemma}

\newtheorem{claim}[theorem]{Claim}
\newtheorem{fact}[theorem]{Fact}

\theoremstyle{definition}

\newtheorem{remark}[theorem]{Remark}

\crefname{equation}{Eqn.}{Eqns.}

\newcommand{\highlight}[1]{\textit{\textbf{#1}}}

\DeclareSymbolFont{bbold}{U}{bbold}{m}{n}
\DeclareSymbolFontAlphabet{\mathbbold}{bbold}

\newcommand{\IGNORE}[1]{}

\newcommand{\poly}{\operatorname{poly}}

\newcommand{\Decomp}{\mathsf{Decomp}}

\newcommand{\kbnote}[1]{{\color{blue}[Koustav: #1]}}
\usepackage[normalem]{ulem} 

\mdfdefinestyle{MyFrame}{%
    roundcorner=10pt,
    innertopmargin=10pt,
    innerbottommargin=10pt,
    innerrightmargin=10pt,
    innerleftmargin=10pt,
    backgroundcolor=gray!5!white}

\title{Near-Optimal Vertex Fault-Tolerant Labels for Steiner Connectivity}
    \author{
    Koustav Bhanja\thanks{Weizmann Institute. Email: \texttt{koustav.bhanja@weizmann.ac.il}. Supported by Merav Parter's European Research Council (ERC) grant under the European Union’s Horizon 2020 research and
    innovation programme, grant agreement No. 949083.}
    \and
    Asaf Petruschka\thanks{Weizmann Institute. Email: \texttt{asaf.petruschka@weizmann.ac.il}.  Supported by an Azrieli Foundation fellowship, and by Merav Parter's European Research Council (ERC) grant under the European Union’s Horizon 2020 research and
    innovation programme, grant agreement No. 949083.}
    }
\date{}

\begin{document}

\maketitle
\pagenumbering{roman}
\input{abstract}
\pagenumbering{arabic}

\input{intro}

\input{approach}

\input{prelim}
\input{basic-tools}

\input{labels}

\section*{Acknowledgments}

We are thankful to Merav Parter for encouraging this collaboration and for valuable discussions.

\bibliographystyle{alphaurl}
{\small \bibliography{references}}

\newpage
\appendix

\input{warm-up-new}

\end{document}

%% file: abstract.tex
\begin{abstract} 

We present a compact labeling scheme for determining whether a designated set of terminals in a graph remains connected after any $f$ (or less) vertex failures occur.
An \emph{$f$-FT Steiner connectivity labeling scheme} for an $n$-vertex graph $G=(V,E)$ with terminal set $U \subseteq V$  provides labels to the vertices of $G$, such that given only the labels of any subset $F \subseteq V$ with $|F| \leq f$, one can determine if $U$ remains connected in $G-F$.
The main complexity measure is the maximum label length.

The special case $U=V$ of global connectivity has been recently studied by Jiang, Parter, and Petruschka~\cite{JiangPP25}, who provided labels of $n^{1-1/f} \cdot \poly(f,\log n)$ bits.
This is near-optimal (up to $\poly(f,\log n)$ factors) by a lower bound of Long, Pettie and Saranurak~\cite{LongPS25}.
Our scheme achieves labels of $|U|^{1-1/f} \cdot \poly(f, \log n)$ for general $U \subseteq V$, which is near-optimal for any given size $|U|$ of the terminal set.
To handle terminal sets, our approach differs from~\cite{JiangPP25}.
We use a well-structured Steiner tree for $U$ produced by a decomposition theorem of Duan and Pettie~\cite{DuanP20},
and bypass the need for Nagamochi-Ibaraki sparsification~\cite{NagamochiI92}.

\end{abstract}

        
        


\IGNORE{
\kbnote{An attempt:}
Let $G=(V,E)$ be an undirected graph on $n=|V|$ vertices and $U\subseteq V$ be a given set of terminal vertices in $G$. 
Steiner connectivity is a fundamental concept that includes global connectivity ($|U|=n$), as well as $(s,t)$-connectivity ($|U|=2$), as just a special case. 
In this paper, we provide a compact labeling scheme for Steiner connectivity that can handle the failure of vertices.
A $f$-vertex fault-tolerant ($f$-VFT) labeling scheme for Steiner connectivity is to assign every vertex with a small label such that, after failure of any set $F\subseteq V$ containing at most $f$ vertices, we can determine whether terminal set $U$ is disconnected in $G-F$ by using only the labels of vertices in $F$. 

There exist several compact $f$-VFT labeling schemes for both global connectivity [References] and (s,t)-connectivity [References]. Unfortunately, no compact $f$-VFT labeling scheme for Steiner connectivity exist till date for any value of $f\ge 2$. Therefore, to address this gap, we present the following first results on Steiner connectivity for any given terminal set $U$ satisfying $2\le |U|\le n$.   

\subparagraph*{$f$-VFT Labeling Scheme for Steiner Connectivity:} There is a labeling scheme that for every vertex $v\in V$, assign a label $L(v)$ occupying $|U|^{1-\frac{1}{f}}.poly(f,\log{n})$ bits of space with the following property. For every $F\subseteq V$ satisfying $|F|\le f$, we can decide if the failure of all vertices in $F$ disconnects $U$ by only inspecting the $L(.)$ labels of $F$.  \\

\noindent
There is a lower bound of $\Omega(|U|^{1-\frac{1}{f}}/f)$ bits of space on the label size in the worst case, which makes our result optimal upto polynomial factors of $f$.
Moreover, our result achieve the same bound on the label size of the existing best-known results for the two extreme scenarios -- global and $(s,t)$-connectivity.
}

%% file: intro.tex
\section{Introduction}

Labeling schemes are distributed graph data structures with many applications in graph algorithms and distributed computing.
The goal is to assign the vertices (and/or edges)
of the graph with succinct yet informative labels, so that queries of interest can be answered from only the labels of the query arguments, without any access to the graph or to any other centralized storage.
The main complexity measure is the \emph{label length}, which is the maximum number of bits stored in a single label; construction and query time are secondary measures.
Diverse types of labeling have been studied, with notable examples being adjacency~\cite{Breuer66,BreuerF67,KannanNR92,AlstrupDK17,AlstrupKTZ19},
ancestry and least common ancestors in trees~\cite{AbiteboulAKMR06,AlstrupHL14},
exact or approximate distances~\cite{ThorupZ05,GavoillePPR04,AlstrupGHP16,BonamyGP22,GawrychowskiU23}, 
and pairwise edge- or vertex-connectivity~\cite{KatzKKP04,HsuL09,IzsakN12,PettieSY22}.

The error-prone nature of many real-life networks has further motivated researchers to consider labeling schemes designed to support events of failures in the graph.
Courcelle and Twig~\cite{CourcelleT07} were the first to explicitly define and study \emph{fault-tolerant (FT)} (a.k.a \emph{forbidden set}) labeling schemes, where the goal is to answer queries subject to the failure (i.e., deletion) of a set $F$ of vertices or edges in the graph.
Here, the set $F$ is part of the query arguments, so the labels assigned to $F$ are available to the query algorithm.
Usually, $|F|$ is assumed to be bounded by some integer parameter $f \geq 1$ (this is abbreviated to the $f$-FT model).
A lot of the works on FT labeling studied special graph families, such as planar or bounded doubling-dimension graphs~\cite{CourcelleT07,CourcelleGKT08,AbrahamCG12,AbrahamCGP16,Bar-NatanCGMW22,BonehCGMW25}.
In this work, we consider a general $n$-vertex (undirected) graph $G = (V,E)$ and focus on connectivity-related $f$-FT labeling schemes.

\paragraph{$s$-$t$ connectivity.}
The ``$s$-$t$ variant'' of this problem has been extensively studied in recent years.
Here, a query consists of a triplet $(s,t,F)$, and it is required to determine if $s,t$ are connected in $G-F$.
The study of $f$-FT $s$-$t$ connectivity labeling was initiated by Dory and Parter~\cite{DoryP21} with subsequent works in~\cite{IzumiEWM23,ParterP22a,ParterPP24,LongPS25}, rendering it reasonably well-understood by now: it is known to admit very succinct labels of only $\poly(f, \log n)$ bits, either with edge failures or with vertex failures; we refer to~\cite{LongPS25} for a detailed presentation of the state-of-the-art bounds.

\paragraph{Global connectivity.}
Parter, Petruschka and Pettie~\cite{ParterPP24} have proposed the ``global connectivity'' variant, where the query consists only of $F$, and it is required to determine if $G-F$ is a connected graph.
Unlike with $s$-$t$ connectivity, the global variant turns out to behave very differently with edge failures ($F \subseteq E$) or with vertex failures ($F \subseteq V)$.
Essentially all known labeling schemes for the $s$-$t$ connectivity variant under edge failures~\cite{DoryP21,IzumiEWM23,LongPS25} also solve the global variant en route, hence the latter also admits succinct $\poly(f,\log n)$-bit labels.
However, in case of vertex failures, it was shown in~\cite{ParterPP24} that polynomial dependency in $f$ is impossible, and their lower bound was improved to $\Omega(n^{1-1/f} /f)$ by Long, Pettie and Saranurak~\cite{LongPS25}.
Very recently, Jiang, Parter and Petruschka~\cite{JiangPP25} gave a nearly matching upper bound by providing a construction of $f$-FT global connectivity labels of $n^{1-1/f} \cdot \poly(f, \log n)$ bits, essentially settling the problem.

\paragraph{Steiner connectivity.}
In this work, we study the natural generalization of the problem to \emph{Steiner connectivity}.
In this variant, the graph $G = (V,E)$ comes with a (fixed) subset of designated vertices $U \subseteq V$ called \emph{terminals}.
Given a query $F \subseteq V \cup E$ of up to $f$ faulty vertices/edges, it is required to report if $F$ is a \emph{$U$-Steiner cut} in $G$, i.e., if there is a pair of disconnected terminals in $G-F$.
Note that $U$ is fixed and not a part of the query, so we only get to see the labels of $F$.%
\footnote{
If we also get the labels of $U$, we can just use the $\poly(f,\log n)$-bit labels of the $s$-$t$ variant, and apply $|U|-1$ pairwise connectivity queries (between one terminal in $U$ to all the others) in $G-F$.
}
We call this variant \emph{$f$-FT Steiner connectivity labeling scheme}.

We first note that the lower bound for the global connectivity variant of~\cite{LongPS25} immediately yields a lower bound for the Steiner setting as well:%
\footnote{Construct the lower bound instance of~\cite{LongPS25} on $k$ vertices set as terminals, and add $n-k$ isolated vertices.}
\begin{theorem}[\cite{LongPS25}]
    Any $f$-FT Steiner connectivity labeling schemes for $n$-vertex graphs with $k$ terminals must have label length of $\Omega(k^{1-1/f}/f)$ bits (in the worst case). 
\end{theorem}

Our technical contribution is providing a nearly matching upper bound which is optimal up to $\poly(f,\log n)$ factors, hence generalizing the result of~\cite{JiangPP25} (recovered when $U=V$):
\begin{theorem}\label{thm:main}
    For every $n$-vertex graph $G = (V,E)$ with terminal set $U \subseteq V$ and $f \geq 1$,
    there is an $f$-FT Steiner connectivity labeling scheme for $(G,U)$ with labels of $|U|^{1-1/f} \cdot \poly(f, \log n)$ bits.
    The labels are computed deterministically in polynomial time.
\end{theorem}

The $f$-FT Steiner connectivity labeling scheme of~\Cref{thm:main} supports any query $F \subseteq V \cup E$ of vertices and edges with $|F| \leq f$.
However, the challenge lies entirely in the case of vertex failures, i.e., $F \subseteq V$: proving~\Cref{thm:main} for this case implies also the general case.
Indeed, we can construct a new instance $(G',U)$ where $G'$ is obtained from $G$ by splitting each edge $e$ of $G$ by a new middle $v_e$, so the failure of $e$ in $G$ is equivalent to the failure of $v_e$ in $G'$ (in terms of affect on connectivity between original vertices).

As for the case of only edge failures $F \subseteq E$, much like in the global connectivity variant, it behaves very differently and admits labels of $\poly(f, \log n)$ bits, implicit by known constructions of $f$-FT connectivity labels for pairwise connectivity queries (the standard $s$-$t$ variant in the literature).
Specifically, both the labels of~\cite{DoryP21} based on linear graph sketches and its derandomized version in~\cite{IzumiEWM23} are straightforward to augment to answer Steiner connectivity queries,%
\footnote{
These are based on fixing some spanning tree $T$ for $G$, and reconnecting the connected components of $T-F$ to obtain those of $G-F$.
Augment the label of each tree edge with the number of terminals in the subtree below it.
Then we can invoke a similar strategy to the sketch computation in~\cite{DoryP21,IzumiEWM23} to determine the number of terminals in each component of $T-F$, and hence also of $G-F$.
}
and other constructions in~\cite{DoryP21,LongPS25} seem plausibly amenable to such augmentations as well.
In light of the discussions above, in the rest of the paper we will only consider the case of vertex failures, i.e., queries are of the form $F \subseteq V$ with $|F| \leq f$.

\paragraph{Related work: Steiner variants.}
We mention that providing ``Steiner generalizations'' that interpolate between the two extreme scenarios of the terminal set ($|U|=n$ and $|U|=2$) is a common objective in various algorithmic and graph-theoretic problems, such as computing Steiner mincut \cite{DBLP:journals/siamcomp/DinitzV00, cole2003fast, he2024cactus, DBLP:conf/esa/0001024}, data structures for Steiner mincuts \cite{dinitz1994connectivity, DBLP:journals/siamcomp/DinitzV00, DBLP:conf/sosa/BaswanaP25} and minimum+1 Steiner cuts \cite{bhanja2025}, sensitivity oracles for Steiner mincuts \cite{DBLP:conf/isaac/Bhanja24, bhanja2025},
Steiner connectivity augmentation and splitting-off \cite{cen2023steiner} and construction of cactus graphs for compactly storing Steiner mincuts \cite{DBLP:journals/siamcomp/DinitzV00, he2024cactus}.

\paragraph{Related work: connectivity oracles.}
FT connectivity labeling schemes can be seen as distributed versions of \emph{connectivity oracles under failures}, which are centralized data structures designed to answer the same kind of queries.
These have been extensively studied especially in the $s$-$t$ variant, with both edge
faults~\cite{PatrascuT07,DuanP20,GibbKKT15} and vertex faults~\cite{DuanP20,BrandS19,LongS22,PilipczukSSTV22,kosinas:LIPIcs.ESA.2023.75,LongW24}.
Much like with labeling schemes, in the global variant of connectivity oracles, vertex failures require much larger complexity than edge failures, as was first observed by Long and Saranurak~\cite{LongS22}.
Very recently, Jiang, Parter and Petruschka~\cite{JiangPP25} provided almost-optimal oracles for global connectivity under vertex failures, and their approach seems plausible to extend also to the Steiner variant.

%% file: approach.tex
\subsection{Our Approach}\label{sect:overview}

While our scheme of~\Cref{thm:main} and its formal analysis are not complicated, the intuitions behind them might seem illusive.
In this section we aim to clarify these.
We first introduce a warm-up suboptimal labeling scheme for the Steiner case, then survey the solution of~\cite{JiangPP25} for the global variant, and finally explain how we ``glue together'' ideas from both to get our solution.
For notational convenience, we use here $\tilde{O}(\cdot)$ to hide $\poly(f,\log n)$ factors.

A basic building block are the previously-introduced  $f$-FT labels for the $s$-$t$ connectivity variant, with length $\tilde{O}(1)$ bits~\cite{ParterPP24,LongPS25}.
Their construction is highly complicated, but we use them as a black-box.
Let $\ell(\cdot)$ denote these labels, and
$L(\cdot)$ the labels that we construct; each $L(\cdot)$-label will store many $\ell(\cdot)$-labels.

\paragraph{Warm-up: using degree classification.}
To gain intuition,
we first sketch a \emph{suboptimal} construction based on a strategy introduced in~\cite{ParterP22a}.
Assume first $U=V$.
Choose an arbitrary spanning tree $T$ in $G$.
Let $r$ be some degree threshold.
Let $B$ be the set of vertices with $T$-degree $\geq r$ (high-deg), and $V-B$ be the rest (low-deg).
For a low-deg $x \in V-B$, let $L(x)$ store $\{\ell(x)\} \cup \{\ell(y) \mid \text{$y$ is $T$-neighbor of $x$} \}$. 
This handles the following queries:
\begin{itemize}
    \item \highlight{(all-low-deg)} If $F \subseteq V-B$, use the $\ell(\cdot)$-labels to check pairwise connectivity between all $T$-neighbors of $F$; one observes that $F$ is a cut in $G$ iff we find some disconnected pair.
\end{itemize}
To handle the remaining queries where $F \cap B \neq \emptyset$, we use recursion.
For each $x \in V$ let $L(x)$ store, for every $y \in B$, $x$'s label \emph{to handle $f-1$ faults in $G-\{y\}$} (recursively).
Then,
\begin{itemize}
    \item \highlight{(one-high-deg)} If there is some $y \in F \cap B$, we answer the query $F-\{y\}$ in $G-y$.
\end{itemize}
Let $b_f$ denote the resulting label length.
As $|B| = O(n/r)$, we get $b_f = \tilde{O}(r + (n/r) \cdot b_{f-1})$.
Optimizing $r$ and solving the recurrence (with base case $b_1 =1$) yields $b_f = \tilde{O}(n^{1-1/2^{f-1}})$.

Generalizing to the Steiner variant with arbitrary terminals $U \subseteq V$ is rather smooth.
Essentially, one takes $T$ as a \emph{minimal Steiner tree} that connects all terminals $U$.
Then high-deg vertices in $T$ are only $|B| \leq O(|U|/r)$, which eventually results in $\tilde{O}(|U|^{1-1/2^{f-1}})$-bit labels.
We give the details in~\Cref{sect:warm-up}.

\paragraph{Existing solution: the global case.}
As previously mentioned, Jiang, Parter and Petruschka~\cite{JiangPP25} gave near-optimal labels for the $f$-FT global connectivity variant, i.e., the case $U=V$.
Roughly speaking, they also use a high/low-deg partition $B$ and $V-B$ (with threshold $r$),
but divide queries $F$ in a dual manner to the warm-up above:
\highlight{(all-high-deg)} $F \subseteq B$, or \highlight{(one-low-deg)} $F \cap (V-B) \neq \emptyset$.
The high/low-deg partition is w.r.t.\ to degrees in $G$, not in a spanning tree.
To enable this, they use \emph{Nagamochi-Ibaraki sparsification}~\cite{NagamochiI92},
yielding a subgraph $H$ of $G$ with $\tilde{O}(n)$ edges such that $G-F$ and $H-F$ have the same connected components, for every $F \subseteq V$, $|F| \leq f$.
So w.l.o.g., $G$ has only $m = \tilde{O}(n)$ edges.

To see the broad strokes, we explain their solution in the special case that $G$ is $f$-vertex-connected.
As~\cite{JiangPP25} observe, in this case, if the query $F$ is indeed a cut, then every $x \in F$ is adjacent to every connected component of $G-F$.
Thus, for a low-deg $x \in V-B$, one lets $L(x)$ store $\{\ell(y) \mid (x,y)\in E\}$, and answers a query with $x \in F$ by looking for a pair of neighbors disconnected in $G-F$.
Note that this argument crucially hinges on the fact that these are all of $x$'s neighbors in the (sparsified) graph $G$, and not in a spanning tree.
To handle queries $F \subseteq B$, one observes that a high-deg $x\in B$ can participate in at most $O(|B|^{f-1})$ such queries, and simply writes the answers to these inside $L(x)$.
As $|B| \leq O(m/r)$, the label length is $\tilde{O}(r + (m/r)^{f-1})$.
Setting $r = m^{1-1/f}$ gives $\tilde{O}( m^{1-1/f}) = \tilde{O}(n^{1-1/f})$.

\paragraph{Our new solution: the Steiner case.}
The main issue that prevents extending the solution of~\cite{JiangPP25} to the Steiner case with arbitrary terminals $U \subseteq V$ lies in the sparsification step.
Intuitively, to get $\tilde{O}(|U|^{1-1/f})$ label length, 
we would like to have $\tilde{O}(|U|)$ edges in the sparsified subgraph $H$,
but then preserving connected components under faults is clearly impossible: $|U|$ can be much smaller than $|V|$, so such $H$ cannot even be connected.
To tackle this, one may wish to construct other meaningful sparsifiers for Steiner connectivity under faults.
This seems to be a challenging task, as even a plausible definition is non-trivial.

Instead, we take a different route and combine the ``best-of-both-worlds'': we use a Steiner tree for high/low-deg partition as in the warm-up, but (roughly speaking) divide queries $F$ to \emph{all-high-deg} or \emph{one-low-deg} as in~\cite{JiangPP25}.
Thus, even in the global case $U=V$, our approach yields a near-optimal scheme which is different than that of~\cite{JiangPP25}.

On a high level, the key is not using just any (minimal) Steiner tree $T$, but one where high-degree vertices $B$ admit a lot of structure, by using a powerful graph decomposition of Duan and Pettie~\cite{DuanP20}.
While it still holds that $|B| \leq O(|U|/r)$, the additional structure is that after deleting $B$ from $G$, the remaining forest $T-B$ is still a Steiner forest for $U-B$
(i.e., connected terminals in $G-B$ are also connected in $T-B$).
The main technical part of our solution lies in a structural analysis that studies the interplay between the (known-in-advance) $G-B$ and $T-B$, and the terminals in the graph $G-F$ defined by a given query $F$ (this is mostly reflected in the correctness proof, specifically \Cref{claim:completness}).
Essentially, this extra structure lets us design labels where it is enough for low-deg vertices to be defined w.r.t.\ the tree, and store the $\ell(\cdot)$-labels of only their tree neighbors.

%% file: prelim.tex
\section{Preliminaries}

Throughout the paper, $G = (V,E)$ is an undirected graph with $n$ vertices, $U \subseteq V$ is a designated terminal set in $G$, and $f \geq 1$ is an integer parameter.

For $K, W \subseteq V$, we say that $K$ \emph{separates} $W$ if there exist two vertices $w,w' \in W-K$ that are disconnected in $G-K$.
A vertex set $F \subseteq V$ is called \emph{a cut} in $G$ if it separates $V$,
and a \emph{Steiner cut} if it separates the terminals $U$.

To shorten statements of theorems, lemmas, etc., we always use the term `algorithm' to mean a deterministic polynomial-time algorithm.

%% file: basic-tools.tex
\section{Basic Tools}\label{sect:basic-tools}

As mentioned in~\Cref{sect:overview}, the most basic tool, which we use as a black box, is the succinct $f$-FT labels for $s$-$t$ connectivity:

\begin{theorem}[\cite{LongPS25}]\label{thm:st-labels}
    There is an algorithm that outputs an $O(f^4 \log^{7.5}(n))$-bit label $\ell(v)$ to every $v \in V$ with the following property:
    
    For every $F \subseteq V$ with $|F| \leq f$ and $s, t \in V$, one can determine if $s,t$ are connected in $G-F$, by only inspecting the $\ell(\cdot)$-labels of $F \cup \{s,t\}$.
\end{theorem}

As explained in~\Cref{sect:overview}, our labels $L(\cdot)$ of~\Cref{thm:main} store many $\ell(\cdot)$ labels, and to answer a query $F$, we make many $s$-$t$ connectivity queries in $G-F$ using these $\ell(\cdot)$-labels.
Suppose that during the query algorithm, we found in this way a pair $s,t$ which is disconnected in $G-F$.
Then clearly $F$ is a cut in $G$, but we still cannot determine that is a \emph{Steiner cut}, as one of $s,t$ might lie in a connected component of $G-F$ without any terminals.
We now present a simple tool that augments~\Cref{thm:st-labels} to help us detect this situation.

\begin{lemma}\label{cor:is-conn-to-U-labels}
    There is an algorithm that outputs an $O(f^4 \log^{7.5})$-bit label $\bar{\ell}(v)$ to every $v \in V$ with the following property:
    
    For every $F \subseteq V$ with $|F| \leq f$ and $x \in V$, one can determine there exists some $u \in U$ such that $x,u$ are connected in $G-F$, by only inspecting the $\bar{\ell}(\cdot)$-labels of $F \cup \{x\}$.
\end{lemma}
\begin{proof}
    Construct an auxiliary graph $G'$ from $G$ by adding a new vertex $z$ and adding an edge $(z,u)$ for every terminal $u \in U$.
    Let $\ell'(\cdot)$ be the labels computed by the algorithm of~\Cref{thm:st-labels} on $G'$, and define $\bar{\ell}(v) := (\ell'(v), \ell'(z))$.
    Now, let $F \subseteq V$ with $|F| \leq f$ and $x \in V$.
    Note that there exists some $u \in U$ which is connected to $x$ in $G-F$ iff $x,z$ are connected in $G'-F$.
    The $\bar{\ell}(\cdot)$-labels of $F \cup \{x\}$ contain all $\ell'(\cdot)$-labels of $F \cup \{x,z\}$, so we can use them to determine if the latter condition holds.
\end{proof}

%% file: labels.tex
\section{The Labeling Scheme: Proof of~\Cref{thm:main}}

In this section, we prove our main~\Cref{thm:main}, giving $f$-FT Steiner connectivity labels of $|U|^{1-1/f} \cdot \poly(f, \log n)$ bits.
Recall that, as discussed in the introduction, we only need to consider up to $f$ vertex failures,
i.e., queries $F \subseteq V$ with $|F| \leq f$.

\paragraph{Subset labels.}
We first introduce a somewhat technical lemma.
Roughly speaking, in terms of the high-level intuition of~\Cref{sect:overview}, it serves to treat the \emph{all-high-degree} case.
The idea is to store explicit information for some relatively small collection of vertex-subsets (of ``high-deg vertices'', see~\Cref{sect:overview}), such we have the information on $K \subseteq F$, we can determine if $K$ itself already separates $U-F$, which implies that $F$ is a Steiner cut.

\begin{lemma}\label{lem:explicit-labels}
    There is an algorithm that given any $K \subseteq V$, outputs a \emph{subset label} $\hat{\ell}(K)$ of $O(f \log n)$ bits with the following property:

    For every $F \subseteq V$ with $|F| \leq f$, one can determine if $K$ separates $U-F$, by only inspecting $\hat{\ell}(K)$ and the names/identifiers of the vertices in $F$. 
\end{lemma}
\begin{proof}
    Let $C_1, \dots, C_p$ be the connected components of $G-K$ that contain some terminal, and denote $U_i = C_i \cap U$, where w.l.o.g.\ $|U_1| \leq \dots \leq |U_p|$.
    Our goal is that given $\hat{\ell}(K)$ and any $F \subseteq V$, $|F|\leq f$, we can decide if at least two sets in $\{U_i-F \mid 1 \leq i \leq p\}$ are non-empty.

    Suppose there exists some $q$ such that $\sum_{i=1}^q |U_i| > f$ and $\sum_{i=q+1}^p |U_i| > f$.
    Then for every $F \subseteq V$ with $|F|\leq f$ there is one non-empty set in $\{U_i -F \mid i \leq q\}$ and another in $\{U_i-F \mid i \geq q+1\}$.
    Thus, we just indicate this situation in $\hat{\ell}(K)$ and always answer YES.
    
    Assume now that no such $q$ exists.
    \begin{itemize}
        \item If $|U_p| > f$:
        Then 
        $\sum_{i=1}^{p-1} |U_i| < f$, and we let  $\hat{\ell}(K)$ store $U_1, \dots, U_{p-1}$ (and indicate that $|U_p| > f$).
        This suffices as for every $F \subseteq V$ with $|F|\leq f$ we know that $U_p - F \neq \emptyset$, so we only need to decide if $\{U_i - F \mid 1 \leq i \leq p-1\}$ has a non-empty set.

        \item If $|U_p| \leq f$: 
        We assert that the sum $\sum_{i=1}^p |U_i|$ is at most $3f$.
        Indeed, if it is more than $f$, let $q$ be maximal such that $\sum_{i=q+1}^p |U_i| > f$.
        Then $\sum_{i=1}^q |U_i| \leq f$, $\sum_{i= q+2}^p |U_i| \leq f$ (by maximality), and $|U_{q+1}| \leq |U_p| \leq f$.
        Thus, we just let $\hat{\ell}(K)$ store all of $U_1, \dots, U_p$.
    \end{itemize}
\end{proof}

\paragraph{The Duan-Pettie Decomposition.}

Our label construction crucially relies on a graph decomposition algorithm of Duan-Pettie~\cite{DuanP20} (as explained in~\Cref{sect:overview}).
This algorithm is a recursive version of the F\"{u}rer-Raghavachari~\cite{FurerR94} algorithm for finding a Steiner tree with smallest possible maximum degree, up to $+1$ error.
The Duan-Pettie decomposition utilizes the structural properties of the F\"{u}rer-Raghavachari tree (rather than its strong approximation guarantee).
In the following, a \emph{Steiner forest} means a forest in the graph where each pair of connected terminals (in the graph) is also connected in the forest.

\begin{theorem}[\cite{DuanP20}]\label{thm:Decomp}
    Let $r \ge 3$ be a \emph{degree threshold}.
    There is an algorithm $\Decomp(G,U,r)$ that outputs a pair $(T,B)$
    such that the following hold:
    \begin{enumerate}
    \item $T$ is a Steiner forest for $U$ in $G$, and $T-B$ is a Steiner forest for $U-B$ in $G-B$.
    \label{prop:Steiner}
    \item The maximum degree in the forest $T-B$ is at most $r$.
    \label{prop:max-deg}
    \item $|B| < |U|/(r-2)$ and $|B\cap U| < |U|/(r-1)$.
    \label{prop:few-bads}
    \end{enumerate}
    The running time of $\Decomp(G,U,r)$ is $O(|U| m\log|U|)$.
\end{theorem}

In the following we assume that $f \geq 2$, as~\Cref{thm:main} is trivial when $f=1$: the label of vertex $x$ just stores a single bit indicating whether $\{x\}$ is a Steiner cut.

\paragraph{Labeling.}
Our first step is to invoke $\Decomp(G,U,r := |U|^{1-1/f} + 2)$ of~\Cref{thm:Decomp}; let $(T,B)$ be the output.
For every $x \in V$, we choose an arbitrary terminal $u_x \in U$ such that $x,u_x$ are connected in $G-B$ (where $u_x := \perp$, i.e., a null value, if there is no such terminal).
Next, we construct the $\ell(\cdot)$-labels and $\bar{\ell}(\cdot)$-labels of~\Cref{thm:st-labels,cor:is-conn-to-U-labels}; for every $v \in V$, define $\ell^* (v) := (v, \ell(v), \bar{\ell}(v) )$.
Finally, the $L(\cdot)$ labels are constructed as follows:

\begin{algorithm}[H]
\caption{Creating the label $L(x)$ of a $x \in V$}\label{alg:label}
\begin{algorithmic}[1]
\State \textbf{store} $\ell^*(x)$
\State \textbf{store} $\ell^*(u_x)$ (or $\perp$ if $u_x = \perp$) \label{line:store-ux}
\State \textbf{store} $\{\ell^*(b) \mid b \in B\}$ \label{line:store-B}
\State \textbf{store} $\hat{\ell}(K=\emptyset)$ of~\Cref{lem:explicit-labels} (i.e., applying this lemma with $K = \emptyset$)
\If{$x \notin B$}
    \State \textbf{store} $\{\ell^*(y) \mid \text{$(x,y)$ is an edge of $T-B$}\}$ \label{line:low-deg}
\Else \Comment{$x \in B$}
    \For{every $K \subseteq B$ s.t.\ $x \in K$ and $|K| \leq f$} \label{line:store-explicit-labels1}
        \State \textbf{store} $K$ along with $\hat{\ell}(K)$ of~\Cref{lem:explicit-labels} \label{line:store-explicit-labels2}
    \EndFor
\EndIf
\end{algorithmic}
\end{algorithm}

The label length analysis is rather straightforward,
given in detail in the following claim.
\begin{claim}\label{claim:label-legnth}
    Each label $L(x)$ generated by~\Cref{alg:label} has $|U|^{1-1/f} \cdot \poly(f,\log n)$ bits, and takes polynomial time to compute.
\end{claim}
\begin{proof}
    By~\Cref{thm:st-labels}, \Cref{cor:is-conn-to-U-labels} and~\Cref{lem:explicit-labels} each $\ell^*(\cdot)$-label or $\hat{\ell}(\cdot)$-label consists of $\poly(f, \log n)$ bits,
    so we show that only $O(|U|^{1-1/f})$ such labels are stored in $L(x)$.
    The first four lines only store $O(|B|)$ such labels, and $|B| \leq |U|/(r-2) = |U|^{1/f}$ by the properties of $\Decomp(G,U,r)$ in~\Cref{thm:Decomp}).
    As we assume $f \geq 2$, $|U|^{1/f} \leq |U|^{1-1/f}$, so this is within budget.
    If $x \notin B$,
    then, line~\ref{line:low-deg} stores $\ell^*(y)$ for each neighbor $y$ of $x$ in $T-B$, which are at most $r = O(|U|^{1-1/f})$ again by~\Cref{thm:Decomp}.
    Next suppose $x \in B$.
    The number of sets $K \subseteq B$ s.t.\ $x \in K$ and $|K| \leq f$ can be bounded by
    \[
    \sum_{k=0}^{f-1} \binom{|B|-1}{k} \leq 1 + \sum_{k=1}^{f-1} \left( \frac{|B| \cdot e}{k} \right)^k
    \leq 1 +  |B|^{f-1} \cdot \sum_{k=1}^{f-1} \left(\frac{e}{k}\right)^k 
    \leq O(|U|^{1-1/f})
    \]
    (we used $\binom{n}{k} \leq (n/k)^k$ for $1 \leq k \leq n$, $\sum_{k=1}^{\infty} (e/k)^k$ is convergent, and $|B| \leq |U|^{1/f}$).
    Thus, lines~\ref{line:store-explicit-labels1} and~\ref{line:store-explicit-labels2} store $O(|U|^{1-1/f})$ $\hat{\ell}(\cdot)$-labels.

    Note that 
    constructing the $\ell^*(\cdot)$ labels for all vertices takes polynomial time by~\Cref{thm:st-labels} and~\Cref{cor:is-conn-to-U-labels}.
    Also, we only need to construct $\hat{\ell}(K)$ labels for sets $K \subseteq B$ such that $x \in B$ and $|K| \leq f$ (or such that $K = \emptyset$), and there are only $O(|U|^{1-1/f})$ such sets $K$ by the arguments above;
    constructing each $\hat{\ell}(K)$ takes polynomial time by~\Cref{lem:explicit-labels}.
\end{proof}

\begin{remark} 
    The subset labels of~\Cref{lem:explicit-labels} are trivial to construct with $O(f^2 \log n)$ bits: with notations from the proof, choose arbitrary $f+1$ sets among $U_1,\dots,U_p$ (or all if $p \leq f$), take $f+1$ terminals from each chosen $U_i$ (or all if $|U_i| \leq f$).
    We've included the more efficient version to demonstrate that the $\poly(f, \log n)$ factors in~\Cref{claim:label-legnth} come from the $f$-FT labels of the $s$-$t$ variant (\Cref{thm:st-labels}).
    These have an $\Omega(f + \log n)$-bit lower bound bits~\cite{ParterPP24}, so even if the latter is matched, the contribution of the subset labels is lightweight.
\end{remark}

\paragraph{Query algorithm.}

The pseudo-code for the query algorithm is given in the following~\Cref{alg:query}, designed to return \emph{YES} if the query $F$ is a Steiner cut, and \emph{NO} otherwise.

\begin{algorithm}[H]
\caption{Answering a query $F \subseteq V$, $|F| \leq f$ from the labels $\{L(x) \mid x \in F\}$}\label{alg:query}
\begin{algorithmic}[1]
\State let $K := F \cap B$
\If{$K$ separates $U-F$ in $G$}\label{line:explicit}
    \Return \emph{YES}
\EndIf
\State let $S := \{v \in V-F \mid \text{$\ell^*(v)$ is stored in $L(x)$ of some $x \in F$}\}$ \label{line:S}
\State let $W := \{w \in S \mid \text{$\exists u \in U$ s.t.\ $u,w$ are connected in $G-F$} \}$ \label{line:W}
\If{$W \neq \emptyset$}
    \State choose an arbitrary $w^* \in W$
    \For{every $w \in W - \{w^*\}$}
        \If{$w^*,w$ are disconnected in $G-F$}\label{line:st-queries}
            \Return \emph{YES}
        \EndIf
    \EndFor
\EndIf
\State \Return \emph{NO}
\end{algorithmic}
\end{algorithm}

We now provide the implementation details:
\begin{itemize}
    \item Line~\ref{line:explicit} is executed using $\hat{\ell}(K)$ of~\Cref{lem:explicit-labels} (and the names of the vertices in $F$, given in their labels).
    If $K = F \cap B$ is non-empty, we can find $\hat{\ell}(K)$ in the label $L(x)$ of some $x \in F \cap B$.
    Otherwise, $\hat{\ell}(K)$ is stored in $L(x)$ for every $x \in F$.

    \item Line~\ref{line:W} is executed by using the $\bar{\ell}(\cdot)$-labels of the vertices in $S \cup F$ (which appear in their $\ell^*(\cdot)$-labels).
    For every $w \in S$, we use $\bar{\ell}(w)$ and $\{\bar{\ell}(x) \mid x\in F\}$ to decide if $w$ is connected to some $u \in U$ in $G-F$, which is possible by~\Cref{cor:is-conn-to-U-labels}.

    \item Line~\ref{line:st-queries} is executed using the $\ell(\cdot)$-labels of $w^*,w$ and $F$ (which appear in their $\ell^*(\cdot)$-labels), which enable to determine if $w^*,w$ are connected in $G-F$ by~\Cref{thm:st-labels}.
\end{itemize}

\paragraph{Correctness.}
We need to show that \emph{YES} is returned if and only if $F$ is a Steiner cut, namely $F$ separates $U$ in $G$.
The soundness direction (`only if') is trivial:

\begin{claim}\label{claim:soudness}
    If the query algorithm (\Cref{alg:query}) returns \emph{YES} then $F$ is a Steiner cut in $G$.
\end{claim}
\begin{proof}
    If \emph{YES} was returned in line~\ref{line:explicit},
    then $F$ has a subset $K$ which separates two terminals in $U-F$, so these are also separated by $F$. Thus, $F$ is a Steiner cut.

    Otherwise, \emph{YES} was returned in line~\ref{line:st-queries},
    so there are $w^*,w\in W$ disconnected in $G-F$.
    By definition of $W$ (line~\ref{line:W} of~\Cref{alg:label}), there exists terminals $u^*,u \in U$ such that $u^*,w^*$ and $u,w$ are connected pairs in $G-F$.
    Thus, $u^*,u$ are two terminals disconnected in $G-F$, so $F$ is a Steiner cut.
\end{proof}

Next, we address the completeness (`if') direction:
\begin{claim}\label{claim:completness}
    If $F$ is a Steiner cut in $G$, then the query algorithm (\Cref{alg:query}) returns \emph{YES}.
\end{claim}

\begin{proof}
    Let us first consider the trivial case when $K = F \cap B$ separates $U-F$ in $G$.
    Then~\Cref{alg:query} returns \emph{YES} in line~\ref{line:explicit} as needed.
    Henceforth, we assume that $F$ satisfies the following property:
    \begin{itemize}
        \item[]
        \begin{description}
            \item[(P1)] $F \cap B$ does not separate $U-F$ in $G$.\label{prop:K-not-separating}
        \end{description}
    \end{itemize}
    Note that given Property~(P1), the only possible place where~\Cref{alg:query} can return \emph{YES} is in line~\ref{line:st-queries}, so we aim to show that this is indeed the case.
    We will crucially use Property~(P1) later on in the proof.

    Let $u,u' \in U$ be two terminals disconnected in $G-F$, which exists as $F$ is a Steiner cut.
    We will show there exists $w,w' \in W$ such that (i) $u,w$ are connected in $G-F$ and (ii) $u',w'$ are connected in $G-F$.
    Let us argue why this suffices to complete the proof.
    First, this clearly implies that $W \neq \emptyset$.
    Second, one of $w,w'$ must be disconnected from $w^*$ in $G-F$, as otherwise $w,w'$ would be connected in $G-F$, so by (i) and (ii) we get that $u,u'$ are also connected in $G-F$, but this is a contradiction.
    Therefore, line~\ref{line:st-queries} of~\Cref{alg:query}, when executed with either $w$ or $w'$ must return \emph{YES} as required.

    We thus focus on showing the existence of $w \in W$ which is connected to $u$ in $G-F$ (and the proof for $w'$ and $u'$ is symmetric).
    Let $C_u$ be the connected component of $u$ in $G-F$.
    Recall the definitions of $S$ and its subset $W$ from lines~\ref{line:S} and~\ref{line:W} of~\Cref{alg:query}.
    Note that every vertex $w \in C_u \cap S$ is automatically in $W$, since it belongs to $S$ and is connected to the terminal $u$ in $G-F$.
    So, our goal is simply to show that $C_u \cap S \neq \emptyset$, or in other words, that the $\ell^*(\cdot)$-label of some vertex in $C_u$ is stored in some $L(x)$ of $x \in F$.

    Again, there is a trivial case: if $C_u \cap B \neq \emptyset$ then we are done, because $\{\ell^*(b) \mid b \in B\}$ is stored in every $L(x)$ of $x \in F$ (line~\ref{line:store-B} of~\Cref{alg:label}), so $B \subseteq S$.
    Thus, we assume that:
    \begin{itemize}
        \item[]
        \begin{description}
            \item[(P2)] The connected component $C_u$ of $u$ in $G-F$ does not intersect $B$.
        \end{description}
    \end{itemize}
    In fact, the inclusion of the $\ell^*(\cdot)$-labels of $B$ in every label $L(x)$ is precisely aimed at enforcing Property~(P2).

    Our next step is to ``capitalize on'' properties (P1), (P2), and the important property
    that $T-B$ is a Steiner forest for $U-B$ in $G-B$ (guaranteed by~\Cref{thm:Decomp}).
    Consider the set $N_G(C_u)$ of all vertices outside $C_u$ that are adjacent to $C_u$ in $G$.
    Note that, as $C_u$ is a connected component in $G-F$, we have $N_G (C_u) \subseteq F$.
    Observe that $N_G (C_u)$ separates $U-F$ in $G$, as $u,u' \in U-F$ with $u \in C_u$ and $u' \notin C_u$.
    Thus, by Property~(P1), it must be that $N_G (C_u) \not\subseteq F \cap B$.
    Hence, there exists some $x \in N_G (C_u) \cap (F-B)$.
    Now, using Property (P2), we see that all vertices in $C_u \cup \{x\}$ are connected in $G-B$, so $u$ is connected to $x$ in $G-B$.
    Recall that the terminal $u_x \in U$ is also (by definition) connected to $x$ in $G-B$.
    (Note that $u_x$ exists since $x$ is connected to the terminal $u$ in $G-B$.)
    Hence, $u,u_x$ are connected in $G-B$, and thus also in $T-B$.
    This is the crux for finishing the proof in the following argument. See~\Cref{fig:completness-proof} for an illustration.
    \begin{figure}[t]
    \centering
    \includegraphics[width=0.8\textwidth]{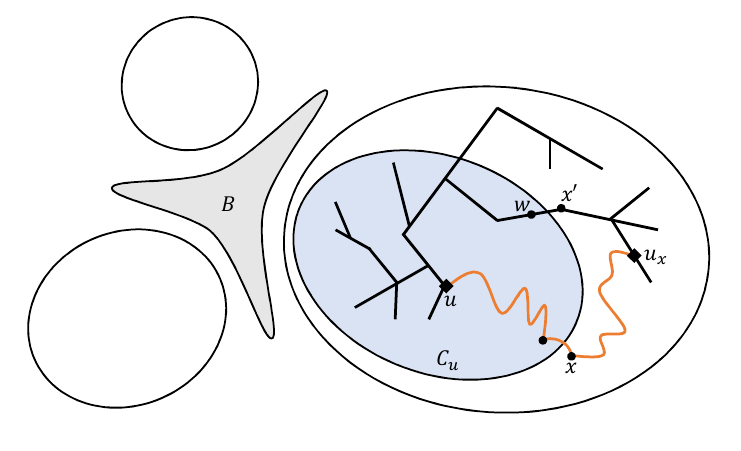}
    \caption{
        Illustrating the proof of~\Cref{claim:completness}.
        Given (P1) and (P2), the connected component $C_u$ of $u$ in $G-F$ is fully contained in some connected component of $G-B$, and $C_u$ has some neighbor $x \in F$ within this bigger component.
        Thus, both $u$ and the chosen terminal $u_x$ for $x$ are connected to $x$ in this component.
        The forest $T-B$ contains a tree that connects \emph{all terminals} in this component, so it must connect $u$ and $u_x$.
        If $u_x \notin C_u$ (as shown in the figure), then the tree must contain a vertex $w \in C_u$ which is a neighbor of some $x' \in F$ in $T-B$, and so $w \in S$.
        Otherwise, $u_x \in C_u \cap S$.
        In either case, we find some vertex in $C_u \cap S$ as required.
    }
    \label{fig:completness-proof}
    \end{figure}

    Observe that $u_x \in S$, since $\ell^*(u_x)$ is stored in $L(x)$ (line~\ref{line:store-ux} of~\Cref{alg:label}).
    So, if $u,u_x$ are connected in $G-F$, then $u_x \in C_u \cap S$ and we are done.
    Otherwise, the path in $T-B$ between $u$ and $u_x$ must go through $F$.
    Let $x' \in F$ be the vertex closest to $u$ on this path, and let $w$ be its neighbor in the direction of $u$ on the path.
    (It might be that $x' = x$.)
    Then the subpath between $u$ and $w$ avoids $F$, so $w \in C_u$.
    Also, $x' \notin B$ and $(x',w)$ is an edge of $T-B$, hence $\ell^*(w)$ is stored in $L(x')$ (line~\ref{line:low-deg} of~\Cref{alg:label}), so $w \in S$.
    Namely, we have shown that $C_u \cap S \neq \emptyset$, and the proof is concluded.    
\end{proof}

\Cref{thm:main} follows immediately from the combination of~\Cref{claim:label-legnth}, \Cref{claim:soudness} and~\Cref{claim:completness}.

%% file: warm-up-new.tex
\section{Warm-up scheme with $\tilde{O}(|U|^{1-1/2^{f-1}})$-bit labels}\label{sect:warm-up}

In this section we provide the details on the warm-up $f$-FT Steiner connectivity scheme with labels of $\tilde{O}(|U|^{1-1/2^{f-1}})$ bits, for an arbitrary terminal set $U \subseteq V$ (as in~\Cref{sect:overview}, $\tilde{O}(\cdot)$ hides $\poly(f,\log n)$ factors).

\paragraph{Steiner tree.}
As mentioned in~\Cref{sect:overview}, the extension from $U=V$ to the arbitrary terminals $U \subseteq V$ is by using a minimal Steiner tree that connects the terminals in $G$, so that the number high-deg vertices is proportional to the of terminals $|U|$ rather than to $n$.
This is by using the following fact:
\begin{fact}
    Let $r \geq 3$.
    Then any tree $T$ with $\ell$ leaves has at most $(\ell-2)/(r-2)$ vertices of degree at least $r$.
\end{fact}
\begin{proof}  
    Let $k$, $p$ and $q$ be the number of vertices with degree $=2$, $\geq 3$ and $\geq r$, respectively.
    Let $s$ be the sum of degrees in $T$.
    On the one hand $s= 2(\ell+k+p -1)$, but on the other $s \geq \ell +2k + 2p + (r-2)q$.
    Rearranging gives $q \leq (\ell-2)(r-2)$. 
\end{proof}

We take a minimal Steiner tree $T$ for $U$ (so its leaves are terminals), and $r \geq 3$ a degree threshold to be set later.
Let $B$ be the set of vertices with degree $\geq r$ in $T$, so $|B| \leq O(|U|/r)$.

\paragraph{Labeling.}
We will the basic tools we need the basic tools from~\Cref{sect:basic-tools}:
the $\ell(\cdot)$-labels $f$-FT labels for $s$-$t$ connectivity of~\Cref{thm:st-labels}, and their extension to $\bar{\ell}(\cdot)$-labels for determining connectivity of a specific vertex to the terminal set in~\Cref{cor:is-conn-to-U-labels}.

As described in~\Cref{sect:overview}, the scheme is recursive.
We use $L_{k} (\cdot, H)$ to denote the labels resulting from it when applied on a subgraph $H$ and fault-parameter $k$ (the terminal set in $H$ is always defined to be $U \cap V(H)$).
We construct the labels $L_f (\cdot, G)$ as follows:

\begin{algorithm}[H]
\caption{Creating the label $L_f(x, G)$ of a vertex $x \in V$}\label{alg : simple label}
\begin{algorithmic}[1]
\State \textbf{store} $\ell(x)$ and $\bar{\ell}(x)$
\For{each $y\in B-\{x\}$}
    \State \textbf{store} $L_{f-1}(x, G-\{y\})$
\EndFor
\If{$x\in V-B$}
    \For{every neighbor $w$ of $x$ in $T$}
        \State \textbf{store} $\ell(w)$ and $\bar{\ell}(w)$
    \EndFor
\EndIf
\end{algorithmic}
\end{algorithm}

\paragraph{Label length.}
Let $b_k$ denote a bound on the bit-length of a label constructed by this scheme with fault parameter $k$ (on any graph with $n$ or less vertices).
As $\ell(\cdot)$-labels and $\bar{\ell}(\cdot)$-labels have $\tilde{O}(1)$ bits, we see that $L_f (x,G)$ has length at most $\tilde{O}(r + b_{f-1}\cdot|B|) = \tilde{O}(r + b_{f-1} \cdot |U|/r)$ bits.
Setting the threshold $r$ to make both terms equals now gives $b_f = \tilde{O}(\sqrt{b_{f-1} |U|})$.
The base case is $b_1 = 1$ (as $L_1 (x, H)$ just needs one bit to specify if there is a pair of disconnected terminals in $H-\{x\}$), and solving the recursion yields $b_f = \tilde{O}(|U|^{1-1/2^{f-1}})$.

\paragraph{Query algorithm.}
Let $F$ be the given query.
\begin{itemize}
    \item \textbf{Case \emph{one-high-deg}:}
    If there is some high-deg $y \in F \cap B$, then every $x \in F-\{y\}$ has stored $L_{f-1} (x, G-\{y\})$ in its label, so we can determine if $F-\{y\}$ is a Steiner cut in $G-\{y\}$, which is equivalent to $F$ being a Steiner cut in $G$.

    \item \textbf{Case \emph{all-low-deg}:}
    If $F \cap B = \emptyset$, then the labels of $F$ store the $\ell(\cdot)$-labels and $\ell(\cdot)$-labels of every $T$-neighbor of $F$.
    Let $W$ be the set of those neighbors.
    Using the $\ell(\cdot)$-labels and $\bar{\ell}(\cdot)$-labels of $W \cup F$, we check for each pair $w,w' \in S$ if (1) $w,w'$ are disconnected in $G -F$ and (2) each of $w,w'$ is connected to some terminal from $U$ in $G-F$.
    We determine that $F$ is a Steiner cut if and only if we find a pair satisfying (1) and (2).

    The correctness of the `if' direction is immediate, as $w,w'$ satisfying (1) and (2) clearly imply a separated pair of terminals in $G-F$.
    For the converse direction, suppose $F$ is a Steiner cut.
    Let $u,u' \in U$ be disconnected in $G-F$,
    and let $T_u$ and $T_{u'}$ be the connected components of $T-F$ containing $u$ and $u'$ respectively.
    Then both $T_u, T_{u'}$ must contain some neighbor of $F$, denoted $w,w'$ respectively.
    Thus $w,u$ and $w',u'$ are connected pairs in $G-F$, hence (2) holds, and (1) then follows since $u,u'$ are disconnected in $G-F$.
\end{itemize}